%% file: 00-Main.tex
\documentclass[sigconf,screen]{acmart}

\AtBeginDocument{%
  }

\copyrightyear{2024}
\acmYear{2024}
\setcopyright{licensedothergov}\acmConference[LICS '24]{39th Annual
ACM/IEEE Symposium on Logic in Computer Science}{July 8--11, 2024}{Tallinn,
Estonia}
\acmBooktitle{39th Annual ACM/IEEE Symposium on Logic in Computer Science
(LICS '24), July 8--11, 2024, Tallinn, Estonia}
\acmDOI{10.1145/3661814.3662111}
\acmISBN{979-8-4007-0660-8/24/07}

\input{_packages}
\input{_macros}

\begin{document}

\title{Decidability of Quasi-Dense Modal Logics} 


\author{Piotr Ostropolski-Nalewaja}
\affiliation{%
  \institution{Technische Universit\"at Dresden}
  \city{Dresden}
  \country{Germany}}
\affiliation{%
\institution{University of Wrocław}
  \city{Wrocław}
  \country{Poland}}
\email{postropolski@cs.uni.wroc.pl}
\orcid{0000-0002-8021-1638}

\author{Tim S. Lyon}
\affiliation{%
  \institution{Technische Universit\"at Dresden}
  \city{Dresden}
  \country{Germany}}
\email{timothy_stephen.lyon@tu-dresden.de}
\orcid{0000-0003-3214-0828}



\begin{abstract}
The decidability of axiomatic extensions of the modal logic $\logick$ with \emph{modal reduction principles}, i.e. axioms of the form $\dia^{k} p \rightarrow \dia^{n} p$, has remained a long-standing open problem. In this paper, we make significant progress toward solving this problem and show that decidability holds for a large subclass of these logics, namely, for \emph{quasi-dense logics}. Such logics are extensions of $\logick$ with with modal reduction axioms such that $0 < k < n$ (dubbed \emph{quasi-density axioms}). To prove decidability, we define novel proof systems for quasi-dense logics consisting of disjunctive existential rules, which are first-order formulae typically used to specify ontologies in the context of database theory. We show that such proof systems can be used to generate proofs and models of modal formulae, and provide an intricate model-theoretic argument showing that such generated models can be encoded as finite objects called \emph{templates}. By enumerating templates of bound size, we obtain an $\algcomplexity$ decision procedure as a consequence.
\end{abstract}




\maketitle

\input{01-Intro}

\input{02-Prelims}

\input{04-FromModalToExistential}

\input{06b-ExistentialRules}

\input{07-Conclusions}

\begin{acks}
Work supported by the European Research Council (ERC) Consolidator Grant 771779 
 (DeciGUT).
\end{acks}
  
\bibliographystyle{abbrv}
\bibliography{bibliography}

\end{document}

%% file: _packages.tex
\usepackage[capitalize]{cleveref}
\usepackage{times}
\usepackage{soul}
\usepackage{url}
\usepackage[utf8]{inputenc}
\usepackage{amsmath}
\usepackage{amsthm}
\usepackage{booktabs}
\usepackage{algorithm}
\usepackage{algorithmic}
\usepackage[switch]{lineno}
\usepackage{blindtext}
\usepackage{multicol}

\newtheorem{theorem}{Theorem}
\newtheorem{corollary}[theorem]{Corollary}
\newtheorem{lemma}[theorem]{Lemma}

\newtheorem{observation}[theorem]{Observation}

\theoremstyle{definition}

\newtheorem{definition}[theorem]{Definition}
\theoremstyle{remark}
\newtheorem{remark}[theorem]{Remark}

\usepackage{comment}
\usepackage{xcolor}
\usepackage{xspace}
\usepackage[normalem]{ulem}

\usepackage{subcaption}
\usepackage[textfont={normalfont, it}]{caption}
\DeclareUnicodeCharacter{2203}{$\exists$}

%% file: _macros.tex
\usepackage{etoolbox}
\makeatletter
\patchcmd{\@bibitem}{\ignorespaces}{\label{bib-#1}\ignorespaces}{}{}
\makeatother

\renewcommand{\vec}[1]{\bar{#1}}
\renewcommand{\phi}{\varphi}

\newcommand{\rset}{\mathcal{R}}
\newcommand{\erule}{\rho}

\newcommand{\apply}[1]{\function{apply}(#1)}

\newcommand{\head}[1]{\function{head}(#1)}
\newcommand{\body}[1]{\function{body}(#1)}

\newcommand{\Prop}{\mathrm{Prop}}
\newcommand{\mlang}{\mathcal{L}_{M}}
\newcommand{\forma}{\phi}
\newcommand{\formb}{\psi}

\newcommand{\dia}{\Diamond}
\newcommand{\logick}{\mathsf{K}}
\newcommand{\logict}{\mathsf{T}}
\newcommand{\logickfour}{\mathsf{K4}}
\newcommand{\logicsfour}{\mathsf{S4}}

\newcommand{\pdset}{\mathcal{P}}
\newcommand{\modw}{W}
\newcommand{\modr}{R}
\newcommand{\modv}{V}
\newcommand{\pdframe}{F(\pdset)}
\newcommand{\pdmodel}{M(\pdset)}
\newcommand{\pdp}{p}

\newcommand{\mforcepd}{\Vdash_{\pdset}}
\newcommand{\logicpd}{\mathsf{L}(\pdset)}
\newcommand{\logic}[1]{\mathsf{L}(#1)}
\newcommand{\md}[1]{d(#1)}
\newcommand{\negnnf}[1]{\dot{\neg}#1}
\newcommand{\mermodels}{\Vdash}

\newcommand{\Exists}[1]{\exists{\scalebox{0.9}{$#1$}}\;}
\newcommand{\Forall}[1]{\forall{\scalebox{0.9}{$#1$}}\;}

\newcommand{\function}[1]{\mathsf{#1}}

\newcommand{\ar}[1]{ar(#1)}

\newcommand{\predicate}[1]{#1} 
\newcommand{\ppred}{\predicate{P}}
\newcommand{\rpred}{\predicate{E}}
\newcommand{\apred}{\predicate{A}}
\newcommand{\bpred}{\predicate{B}}
\newcommand{\cpred}{\predicate{C}}

\newcommand{\epred}{\predicate{R}}
\newcommand{\fpred}{\predicate{F}}

\newcommand{\spred}{\predicate{S}}

\newcommand{\ourterms}{\mathrm{T}}

\newcommand{\constants}{\mathrm{C}}
\newcommand{\variables}{\mathrm{V}}

\newcommand{\signature}{\Sigma}
\newcommand{\pred}{\mathrm{Pred}}
\newcommand{\sig}{\signature}

\newcommand{\database}{\mathcal{D}}
\newcommand{\db}{\database}
\newcommand{\structure}[1]{\mathcal{#1}}
\newcommand{\instance}{\structure{I}}
\newcommand{\inst}{\instance}
\newcommand{\ppath}{P}
\newcommand{\jnstance}{\structure{J}}
\newcommand{\jnst}{\jnstance}

\newcommand{\set}[1]{\{\,#1\,\}}
\newcommand{\pair}[1]{\langle\, #1 \,\rangle}

\newcommand{\Vx}{\vec{x}}

\newcommand{\Vy}{\vec{y}}
\newcommand{\Vz}{\vec{z}}

\newcommand{\vy}{\vec{y}}

\newcommand{\vz}{\vec{z}}

\newcommand{\vt}{\vec{t}}

\newcommand{\chasesymbol}{Ch}

\newcommand{\chase}[1]{\chasesymbol(#1)}

\newcommand{\ruleset}{\mathcal{R}}
\newcommand{\rs}{\ruleset}

\newcounter{rscounter}
\newcounter{dbcounter}
\definecolor{lightblue}{rgb}{0.7, 0.7, 1}
\definecolor{lightred}{rgb}{1, 0.7, 0.7}

\newcommand{\piotr}[1]{\textcolor{lightred}{#1}}

\newcommand{\orange}[1]{{\color{orange} #1}}

\newcommand{\iffi}{\textit{iff} } 

\newcommand{\narule}[3]{#1 \to \Exists{#2}\; #3}
\newcommand{\nadrule}[2]{#1 \to #2}

\newcommand{\arule}[3]{#1 \;\;&\to\;\;\Exists{#2}\; #3}
\newcommand{\adrule}[2]{#1 \;\;&\to\;\;#2}

\newcommand{\theform}{\phi}
\newcommand{\core}{\function{core}}
\newcommand{\unravel}{\function{unravel}}
\newcommand{\unfold}{\function{unfold}}
\newcommand{\labels}{\function{labels}}
\newcommand{\trim}{\function{trim}}

\newcommand{\word}{\function{word}}
\newcommand{\empstr}{\varepsilon}
\newcommand{\embed}{\function{embed}}
\newcommand{\subform}{\function{ms}}

\newcommand{\map}{\mathcal{N}}
\newcommand{\axioms}{\pdset}
\newcommand{\rsp}{\rs_\theform}
\newcommand{\rsa}{\axioms} 

\newcommand{\chases}{\chase{\dbz, \rsp \cup \rsa}}
\newcommand{\ch}{\mathcal{C}}
\newcommand{\dbz}{\db_{\phi}}

\newcommand{\kp}{k_+}
\newcommand{\kkp}{k \rightarrow k_+}

\newcommand{\mlbox}{\Box}
\newcommand{\mldia}{\dia}
\newcommand{\dep}[2]{\mathtt{dep}_{#1}(#2)}
\newcommand\sbullet[1][.5]{\mathbin{\vcenter{\hbox{\scalebox{#1}{$\bullet$}}}}}

\newcommand{\predcontr}{\mathsf{contr}}
\newcommand{\algcomplexity}{\textsc{ExpSpace}}
\newcommand{\assign}{\iota}

%% file: 01-Intro.tex
\section{Introduction}\label{sec:intro}

Modal logic is a rich sub-discipline of mathematical logic, having important applications in numerous areas from computer science to legal theory. For example, temporal logics have been used in the verification of programs~\cite{AlpSch85}, epistemic logics have found use in distributed systems~\cite{MeyMeyHoe04}, and multi-agent logics have been used to analyze legal concepts~\cite{Bro11}. Modal logics are typically obtained by extending propositional logic with a set of modalities, which prefix formulae and qualify the truth of a proposition. In this paper, we focus on a class of normal modal logics~\cite{BlaRijVen01}, dubbed \emph{quasi-dense (modal) logics}, whose language extends propositional classical logic with the possibility modality $\dia$ and the necessity modality $\Box$.


Quasi-dense logics extend the minimal normal modal logic $\logick$ (cf. Blackburn et al.~\cite{BlaRijVen01}) with a finite set of \emph{quasi-dense axioms} of the form $\dia^k p \rightarrow \dia^n p$ such that $0 < k < n$.\footnote{The name `quasi-dense' is obtained from the fact that such axioms are canonical for a generalization of the density property. In particular, the axiom $\dia^k p \rightarrow \dia^n p$ is canonical for the following property: if a directed $R$ path of length $k$ exists from a world $w$ to $u$ in a Kripke model $M = (W,R,V)$, then a directed $R$ path of length $n$ exists from $w$ to $u$ as well. When $0 < k < n$, this implies the existence of $n-1$ points `between' $w$ and $u$ when an $R$ path of length $k$ exists, thus generalizing the usual density property.} This class of axioms forms a prominent subclass of the so-called \emph{modal reduction principles}~\cite{Ben76}, which are modal axioms of the same form, though without the restriction that $k < n$. Logics extending $\logick$ with modal reduction principles were studied at least as far back as 1976 by van Bentham~\cite{Ben76}. The decidability of such logics has remained a longstanding open problem (see~\cite{BlaBenWol07} for a discussion), though partial solutions have been obtained.
%
For example, it is well-known that the decidability of $\logict$, $\logickfour$, $\logicsfour$, and $\logick \oplus \dia p \rightarrow \dia \dia p$ can be obtained via filtration techniques~\cite{BlaRijVen01} and each logic $\logick \oplus \dia^n p \rightarrow \dia p$ is known to be decidable as well (cf.~\cite{TiuIanGor12}). Notable however, is the result of Zakharyaschev~\cite{Zak92}, who showed that any extension of $\logickfour$ with modal reduction principles is decidable. In this paper, we obtain a new substantial result, showing all quasi-dense logics decidable via a novel model-theoretic argument. Thus, we provide additional tools for approaching the general problem and make further headway toward a complete solution.


Our approach to decidability is unique to the literature on modal logics in that we rely on tools from database theory to make our argument. In particular, we rely on \emph{disjunctive existential rules}---normally used in ontology specification---which are first-order formulae of the form $\Forall{\Vx, \Vy} \beta(\Vx, \Vy) \to \bigvee_{1 \leq i \leq n} \Exists{\Vz_{i}} \alpha_{i}(\Vy_{i}, \Vz_{i})$ such that $\beta(\Vx, \Vy)$ and each $\alpha_{i}(\Vy_{i}, \Vz_{i})$ are conjunctions of atomic formulae over constants and the variables $\Vx, \Vy$ and $\Vy_{i}, \Vz_{i}$, respectively~\cite{disjunctive-chase-intro}. In the context of database theory, a finite collection $\ruleset$ of such rules constitutes an \emph{ontology}, while a finite collection $\db$ of first-order atoms over constants constitutes a \emph{database}. Together, $(\db,\ruleset)$ form a \emph{knowledge base}, which may be queried to extract implicit information from the explicitly specified knowledge. The primary tool in querying a knowledge base $(\db,\ruleset)$ is the so-called \emph{(disjunctive) chase}, which is an algorithm that exhaustively applies rules from $\ruleset$ to $\db$, ultimately resulting in a selection of models of the knowledge base to which queries can be mapped~\cite{disjunctive-chase-intro,disjunctive-chase-termination}. 

In this paper, we rely on the above tools, representing a modal formula $\phi$ as a database $\dbz$ and encoding a tableau calculus with a set $\ruleset$ of disjunctive existential rules. Running the disjunctive chase over $(\dbz,\ruleset)$ simulates a tableau algorithm, effectively yielding a tableau proof or model of $\phi$. We find this approach to be preferable to utilizing a standard tableau approach for a couple of reasons: first, this saves us from introducing the various notions required to specify a tableau calculus, which only forms a minor part of our overall decidability argument. Second, the tools used in the context of (disjunctive) existential rules are already perfectly suited for the intricate model-theoretic arguments we employ. To put it succinctly, this approach is more parsimonious than it would otherwise be.

We note that the product of the disjunctive chase, which we denote by $\chase{\dbz,\ruleset}$, consists of a potentially infinite set of \emph{instances}, i.e. potentially infinite sets of atomic first-order formulae, which intuitively correspond to maximal branches in a tableau. When all such  instances contain a contradictory formula $\predcontr$, $\chase{\dbz,\ruleset}$ witnesses the unsatisfiability of $\phi$, and when at least one such instance omits $\predcontr$, that instance can be transformed into a Kripke model witnessing the satisfiability of $\phi$. Since $\chase{\dbz,\ruleset}$ can be infinitely large, the central issue that needs to be overcome to obtain decidability is the `finitization' of $\chase{\dbz,\ruleset}$. Indeed, we show that an instance $\ch \in \chase{\dbz,\ruleset}$ is free of $\predcontr$ \iffi a finite encoding of the model $\ch$ exists, which we refer to as a \emph{template}. As only a finite number of possible templates exist, we obtain a $\algcomplexity$ decision procedure for each quasi-dense logic as a corollary.


This paper is organized as follows: In \cref{sec:prelims}, we define the language and semantics of quasi-dense logics, preliminary notions for disjunctive existential rules, and the various model-theoretic tools required for our decidability argument. In \cref{sec:modal-ER}, we show that the disjunctive chase constitutes a sound and complete `proof system' for quasi-dense logics. In \cref{sec:erules}, we employ an intricate model-theoretic argument showing a correspondence between models generated by the disjunctive chase and finite encodings of such models (i.e. templates), yielding a $\algcomplexity$ decision procedure for quasi-dense logics. Last, in \cref{sec:conclusions}, we share some concluding remarks and discuss future work.

%% file: 02-Prelims.tex
\newcommand{\tree}{\mathcal{T}}

\section{Preliminaries}\label{sec:prelims}

\subsection{Modal Logic}

We let $\Prop := \{p, q, r, \ldots\}$ be a set of \emph{propositional atoms} (which may be annotated), and define the \emph{modal language} $\mlang$ to be the smallest set of \emph{modal formulae} generated by the following grammar: 
$$
\forma ::= p \ | \ \neg p \ | \ \forma \lor \forma \ | \ \forma \land \forma \ | \ \dia \phi \ | \ \Box \phi
$$
where $p$ ranges over $\Prop$. The negation $\negnnf{\phi}$ of a formula $\phi$ is defined recursively as shown below: 
\vspace*{-1em}
\begin{multicols}{2}
\begin{itemize}
    \item $\negnnf{p} := \neg p$;
    
    \item $\negnnf{\neg p} := p$;
    
    \item $\negnnf{(\psi \lor \chi)} := \negnnf{\psi} \land \negnnf{\chi}$;
    
    \item $\negnnf{(\psi \land \chi)} := \negnnf{\psi} \lor \negnnf{\chi}$;
    
    \item $\negnnf{\dia \psi} := \Box \negnnf{\psi}$;
    
    \item $\negnnf{\Box \psi} := \dia \negnnf{\psi}$.
\end{itemize}
\end{multicols}
\vspace*{-1em}

We define the \emph{modal depth} $\md{\phi}$ of a modal formula $\phi$ recursively as follows: (1) $\md{p} = \md{\neg p} = 0$, (2) $\md{\psi \circ \chi} = \max\{\md{\psi}, \md{\chi}\}$ for $\circ \in \{\lor, \land\}$, and (3) $\md{\triangledown \psi} = \md{\psi} + 1$ for $\triangledown \in \{\dia, \Box\}$. Given a modal formula $\forma$, define the set $\subform_i(\forma)$ of $\phi$'s \emph{depth $i$ modal sub-formulae} as follows: $\psi \in \subform_i(\forma)$ \iffi $\psi$ is a sub-formula of $\phi$, $\psi$ is of the form $\triangledown \chi$ with  $\triangledown \in \{\dia, \Box\}$, and $\md{\phi} - \md{\psi} = i$.

Modal formulae are interpreted over special kinds of Kripke models, which we call \emph{$\pdset$-models}, defined below. We note that $\pdset$ is taken to be a finite set of first-order formulae called \emph{quasi-density properties}. A \emph{quasi-density property} (or, \emph{QDP} for short) is a first-order formula $\pdp = \forall x,y (\modr^{k}(x,y) \rightarrow \modr^{n}(x,y))$ such that $0 < k < n$ and where $\modr^{\ell}(x,y) := (x = y)$ if $\ell = 0$ and $\modr^{\ell+1} := \exists z (\modr^{\ell}(x,z) \land \modr(z,y))$. We also use $k \to k_{+}$ to denote a QDP of the above form, where $n = k_{+}$ suggests that $k_{+}$ is larger than $k$. This notation will be used when $k$ should be made explicit and is relevant to the discussion.

\begin{definition}[$\pdset$-Model]\label{def:kripke-model}  Let $\pdset$ be a finite set of QDPs. A \emph{$\pdset$-frame} is defined to be an ordered pair $\pdframe := (\modw,\modr)$ such that $\modw$ is a non-empty set of points, called \emph{worlds}, and the \emph{accessibility relation} $\modr \subseteq \modw \times \modw$ satisfies every $\pdp \in \pdset$. A \emph{$\pdset$-model} is defined to be a tuple $\pdmodel = (\pdframe,\modv)$ such that $F$ is a $\pdset$-frame and $\modv : \Prop \to 2^{\modw}$ is a \emph{valuation function} mapping propositions to sets of worlds.
\end{definition}

\begin{definition}[Semantic Clauses]\label{def:modal-semantics} Let $\pdmodel = (\modw,\modr,\modv)$ be a $\pdset$-model. We define a forcing relation $\mforcepd$ such that
\begin{itemize}

\item $\pdmodel, w \mforcepd p$ \iffi $w \in \modv(p)$;

\item $\pdmodel, w \mforcepd \neg p$ \iffi $w \not\in \modv(p)$;

\item $\pdmodel, w \mforcepd \forma \lor \formb$ \iffi $\pdmodel, w \mforcepd \forma$ or $\pdmodel, w \mforcepd \formb$;

\item $\pdmodel, w \mforcepd \forma \land \formb$ \iffi $\pdmodel, w \mforcepd \forma$ and $\pdmodel, w \mforcepd \formb$;

\item $\pdmodel, w \mforcepd \dia \forma$ \iffi $\exists u \in \modw$, $\modr(w,u)$ and $\pdmodel, u \mforcepd \forma$;

\item $\pdmodel, w \mforcepd \Box \forma$ \iffi $\forall u \! \in \! \modw$, if $\modr(w,u)$, then $\pdmodel, u  \! \mforcepd  \! \forma$;

\item $\pdmodel \mforcepd \forma$ \iffi $\forall w \in \modw$, $\pdmodel, w \mforcepd \forma$.

\end{itemize}
A modal formulae $\forma \in \mlang$ is \emph{$\pdset$-valid}, written $\mforcepd \forma$, \iffi for all $\pdset$-models $\pdmodel$, $\pdmodel \mforcepd \forma$. We define the modal logic $\logicpd \subseteq \mlang$ to be the smallest set of $\pdset$-valid formulae, and note that $\logic{\emptyset}$ is the well-known minimal, normal modal logic $\logick$. 
\end{definition}

In this paper, we prove that each logic $\logicpd$ is decidable, that is, we provide a computable function $f$ such that for any quasi-dense logic $\logicpd$ and modal formula $\phi \in \mlang$, $f(\logicpd, \phi) = 1$ if $\phi \in \logicpd$ and $f(\logicpd, \phi) = 0$ if $\phi \not\in \logicpd$. We note that this provides a partial answer to a $\tilde{40}$ year old problem, first discussed by Wolter and Zakharyaschev~\cite[Problem 8]{BlaBenWol07}. Therefore, the main theorem we demonstrate is the following:

\begin{theorem}\label{thm:main} Given a logic $\logicpd$ and a modal formula $\forma$, it is decidable to check if $\phi \in \logicpd$.
\end{theorem}

\subsection{Existential Rules and Model-Theoretic Notions}

\noindent
\textbf{Formulae and Syntax.} We let $\constants$ and $\variables$ be two disjoint, denumerable sets of \emph{constants} and \emph{variables}. We use $a, b, c, \ldots$ to denote constants, and use $x, y, z, \ldots$ to denote variables (potentially annotated). We define the set of \emph{terms} $\ourterms$ to be the set of words over $\constants \cup \variables$ including the empty word $\empstr$;\footnote{Although this notion of term is non-standard, it is tailored to ease our proofs later on (see \cref{def:unfold} in particular).} we denote terms by $t$, $s$, $\ldots$ (potentially annotated). We use $\sbullet$ to denote the concatenation operation over terms.


Moreover, we let $\pred := \{\predcontr\} \cup \{\apred, \bpred, \cpred, \ldots\}$ be a denumerable set of \emph{predicates} such that $\predcontr$ is a predicate of arity $0$ and the remaining predicates are of arity $1$ and $2$. The nullary predicate $\predcontr$ serves a special purpose in our work and will be used to encode contradictory information (see~\cref{sec:modal-ER}). We use $\ar{\apred} = n$ to denote that $\apred \in \pred$ is of arity $n \in \{0,1,2\}$. An \emph{atom} is a formula of the form $\predcontr$, $\apred(t)$, or $\bpred(t,t')$ such that $t, t' \in \ourterms$, $\ar{\apred} = 1$, and $\ar{\bpred} = 2$. We will often write atomic formulae as $\apred(\vec{t})$ with $\vec{t}$ denoting a single term $t$ or a pair of terms $t,t'$.

An instance $\inst$ is defined to be a (potentially infinite) set of atoms. We use $\inst$, $\jnst$, $\ldots$ (potentially annotated) to denote instances. We define $\apred \in \labels_{\inst}(t)$ \iffi $\apred(t) \in \inst$. When $\inst$ is clear from the context, we simply write $\labels(t)$ omitting the subscript $\inst$. We define a \emph{path} in an instance $\inst$ from a term $t$ to a term $t'$ to be a sequence of binary atoms $\ppath(t,t') = \apred_{1}(s_{0},s_{1}), \ldots, \apred_{n}(s_{n-1}, s_{n}) \in \inst$ for $n > 0$, where $t = s_0$ and $t' = s_n$. 
 We define the \emph{length} of a path to be equal its cardinality. Given a path $\ppath(t,t')$ of the above form, we define the \emph{word of the path} to be $\word(\ppath(t,t')) = s_{0}\cdots s_{n-1}$; note that we omit the end term of the path in its word. 
 An instance $\inst$ is a \emph{directed acyclic graph} (DAG) \iffi it is free of cycles, i.e. there exists no path from a term $t$ in $\inst$ to itself. A DAG $\inst$ is \emph{rooted} at a term $t$ \iffi for every other term $t'$ of $\inst$ there exists a path from $t$ to $t'$. A \emph{tree} is a rooted DAG such that each path from its root to any other term is unique. We define a \emph{multi-tree} to be like a tree, but allowing for multiple edges per pair of terms. We remark that every multi-tree is a rooted DAG. We use $\tree$ and annotated versions thereof to denote (multi-)trees, and given a tree $\tree$ and a term $t$ therein, we let $\tree_t$ denote the sub-tree of $\tree$ rooted at $t$. Moreover, we will employ standard graph-theoretic notions; e.g. the notion of a \emph{descendant} and the \emph{depth} of a tree. Given a rooted DAG $\inst$ with source $t$ containing a term $s$, we define the \emph{depth} of $s$ in $\inst$, denoted $\dep{\inst}{s}$, to be equal to the minimal length among all paths from $t$ to $s$ in $\inst$.\\

\noindent
\textbf{Substitutions and Homomorphisms.} 
We define a \emph{substitution} to be a partial function over $\ourterms$. A \emph{homomorphism} from an instance $\inst$ to an instance $\jnst$ is a substitution $h$ from the terms of $\inst$ to the terms of~$\jnst$ such that (1) if $\apred(t_1, \ldots, t_n) \in \inst$, then $\apred(h(t_1), \ldots, h(t_n)) \in \jnst$, and (2) $h(a) = a$, for each constant~$a \in \constants$. 
A homomorphism $h$ from an instance $\inst$ to $\jnst$ is an \emph{isomorphism} \iffi it is bijective and $h^{-1}$ is also a homomorphism. 
A homomorphism $h: \inst \rightarrow \jnst$ is a \emph{partial isomorphism} \iffi it is an isomorphism from $\inst$ to a sub-instance of $\jnst$ induced by the terms of $h(\inst)$.

We define the \emph{core} (up to isomorphism) of a finite instance $\inst$ to be the instance $\jnst = \core(\inst)$ such that $\jnst \subseteq \inst$, $\inst$ can be homomorphically mapped to $\jnst$, and if there exists a homomorphism $h$ from $\jnst$ to itself then $h$ is an isomorphism. We also define the core of a tree $\tree$ of finite depth in the same manner.\footnote{We remark that cores of infinite instances do not exist in general~\cite{HellNes92}, despite existing for infinite trees of finite depth.}\\ 

\noindent
\textbf{Existential Rules.} A \emph{disjunctive existential rule} is a first-order formula of the form:
$$
\erule = \Forall{\Vx, \Vy} \scalebox{1.5}{$($}\; \beta(\Vx, \Vy) \to \bigvee_{1 \leq i \leq n} \Exists{\Vz_{i}} \alpha_{i}(\Vy_{i}, \Vz_{i}) \;\scalebox{1.5}{$)$}
$$
such that $\Vy = \Vy_{1}, \ldots, \Vy_{n}$ and $\beta(\Vx, \Vy) = \body{\erule}$ (called the body) is a conjunction of atoms over constants and the variables $\Vx, \Vy$, and $\bigvee_{1 \leq i \leq n} \alpha_{i}(\Vy_{i}, \Vz_{i}) = \head{\erule}$ (called the head) is a disjunction such that each $\alpha_{i}(\Vy_{i}, \Vz_{i})$ is a conjunction of atoms over constants and the variables $\Vy_{i}, \Vz_{i}$. As usual, we will often treat $\beta(\Vx, \Vy)$ and $\alpha_{i}(\Vy_{i}, \Vz_{i})$ as sets of atoms, rather than conjunctions of atoms; cf.~\cite{BagLecMugSal11}. 
%
%
We call a finite set $\rset$ of disjunctive existential rules a \emph{rule set} and refer to disjunctive existential rules as rules for simplicity. A \emph{disjunctive datalog rule} is an rule without existential quantifiers. A rule is \emph{non-disjunctive} \iffi its head contains a single disjunct.\footnote{Non-disjunctive rules are also referred to as \emph{tuple-generating dependencies} (TGDs)~\cite{AbiteboulHV95}, \emph{conceptual graph rules}~\cite{SalMug96}, Datalog$^\pm$~\cite{Gottlob09}, and \emph{$\forall \exists$-rules}~\cite{BagLecMugSal11} in the literature, though the name \emph{existential rule} is most commonly used.}\\

\noindent
\textbf{Triggers and Rule Applications.}
Given an instance $\inst$, a rule $\erule$, and a homomorphism $h$ from $\body{\erule}$ to $\inst$ we call the pair $\pair{\erule, h}$ a \emph{trigger} in $\inst$. A trigger is \emph{active} if there exists no disjunct $\gamma$ of $\head{\erule}$ such that $h$ can be extended to map $\body{\erule} \land \gamma$ to $\inst$.

Given an active trigger $\pi = \pair{\erule, h}$ over some instance $\inst$, we define the \emph{application} of the trigger $\pi$ to the instance $\inst$ as the following set of instances:
$$\set{\inst \;\cup\; \alpha(h(\vy), \vt) \;\;\mid\;\; \Exists{\vz} \alpha(\vy, \vz) \text{ is a disjunct of } \head{\erule}}$$
where each $\vt$ is a tuple of fresh terms. We denote an application of a trigger to an instance by $\apply{\inst, \pi}$.\\ 

\noindent
\textbf{Disjunctive Chase.} We use a variant of the restricted disjunctive chase presented by Carral et al. \cite{disjunctive-chase-termination}.
Given an instance $\inst$ and a set of rules $\rs$ we define 
a \emph{fair derivation sequence} of $\inst$ and $\rs$ to be any (potentially infinite) sequence $\set{\inst_i}_{i \in \mathbb{N}}$ satisfying the following:
\begin{itemize}
    \item $\inst_0 = \inst$;
    \item there exists a trigger $\pi_i$ such that $\inst_{i+1} \in \apply{\inst_i, \pi_i}$, for every $i \in \mathbb{N}$;
    \item there exists no trigger $\pi$ that is active indefinitely.
\end{itemize}

Given an instance $\inst$ and a set of rules $\rs$ we define the \emph{chase} $\chase{\inst, \rs}$ as the set containing $\bigcup_{i \in \mathbb{N}}{\jnst_i}$ for every fair derivation sequence $\set{\jnst_i}_{i \in \mathbb{N}}$ of $\inst$ and $\rs$.\\ 

\noindent
\textbf{Formulae and Semantics.} We define a \emph{formula} $X$ to be an expression of the following form: $X: = \predcontr \ | \ \apred(\vec{t}) \ | \ X \lor X \ | \ X \land X$. The \emph{complexity} $|X|$ of a formula $X$ is recursively defined as follows: $|\predcontr| = |\apred(\vec{t})| = 0$ and $|X \circ Y| = |X| + |Y|+ 1$ for $\circ \in \{\lor, \land\}$. We let $\models$ denote the typical first-order entailment relation, which is defined between instances, rule sets, and formulae. Note that $\inst \models \ruleset$ \iffi there are no active $\rs$-triggers in $\inst$, and $\chase{\inst, \rs} \models X$ \iffi for each $\inst \in \chase{\inst, \rs}$, $\inst \models X$

\begin{observation}\label{obs:chase-preserves-homs}
For any two instances $\inst, \jnst$ and a set $\rs$ of non-disjunctive existential rules, if $\inst$ homomorphically maps to $\jnst$, then $\chase{\inst, \rs}$ homomorphically maps to  $\chase{\jnst, \rs}$.
\end{observation}

%% file: 04-FromModalToExistential.tex
\section{From Modal Logic to Rules}\label{sec:modal-ER}

For the remainder of the paper, let us fix a set $\axioms$ of quasi-density properties and a modal formula $\theform$. We let the signature $\sig$ be the set $\set{\ppred_{\psi} \mid \psi \text{ is a sub-formula of } \theform} \cup \set{\epred, \rpred, \predcontr}$ such that each predicate $\ppred_{\psi}$ is unary, and both $\epred$ and $\rpred$ are binary. We will often refer to unary atoms as \emph{labels} of terms. Below, we define a set of disjunctive existential rules that effectively function like a tableau system for $\phi$.\footnote{The relationship between existential rules and sequent systems (which are well-known to be equivalent to tableaux) has been studied previously~\cite{LyoOst22}.} This set is comprised of two distinct parts: the QDPs in $\axioms$ and a set $\rsp$ that encodes the semantic clauses from \cref{def:modal-semantics}. By standard first-order equivalences, every QDP of the form $\forall x,y (\modr^{k}(x,y) \rightarrow \modr^{n}(x,y))$ is equivalent to a non-disjunctive existential rule of the form:
$$
\forall x, \Vec{y}, y R(x,y_{1}) \land \cdots \land R(y_{k-1},y) \rightarrow \exists \Vec{z} R(x,z_{1}) \land \cdots \land R(z_{n-1},y)
$$
where $\Vec{y} = y_{1}, \ldots, y_{k}$ and $\Vec{z} = z_{1}, \ldots, z_{n}$. Therefore, we are permitted to treat $\axioms$ as a rule set. The second rule set $\rsp$, which encodes the modal semantics, is defined as follows:

\begin{definition}\label{def:rsp}
    Let $\rsp$ be the rule set consisting of the following:
    \begin{align*}
        \adrule{\ppred_{\psi \land \chi}(x)}{\ppred_{\psi}(x) \land \ppred_{\chi}(x)}\tag{i}\label{rule:head-conjunction}\\
        \adrule{\ppred_{\psi \lor \chi}(x)}{\ppred_{\psi}(x) \lor \ppred_{\chi}(x)}\tag{ii}\label{rule:head-disjunction}\\
        \adrule{\ppred_{\mlbox\psi}(x) \land \epred(x, y)}{\ppred_{\psi}(y)}\tag{iii}\label{rule:dia}\\
        \arule{\ppred_{\mldia\psi}(x)}{y}{\epred(x,y) \land \rpred(x, y) \land \ppred_{\psi}(y)}\tag{iv}\label{rule:box}\\
        \adrule{\ppred_{\psi}(x) \land \ppred_{\negnnf{\psi}}(x)}{\predcontr}\tag{v}\label{rule:contr}
    \end{align*}
    for $\psi, \chi$ ranging over sub-formulae of $\theform$.
\end{definition}

The above rule set encodes the the semantic clauses from \cref{def:modal-semantics} in the following way: Rules (i) and (ii) encode the fact that if $\psi \land \chi$ or $\psi \lor \chi$ is satisfied at a world $x$ in a Kripke model, then $\phi$ and (or) $\chi$ is satisfied at $x$. Rule (iii) encodes the fact that if $\Box \psi$ is satisfied at a world $x$, which relates to a world $y$ via the accessibility relation $R$, in a Kripke model, then $\psi$ holds at $y$. Rule (iv) captures the fact that if $\dia \psi$ is satisifed at $x$ in a Kripke model, then there exists a world $y$ such that $R(x,y)$ with $\psi$ holding at $y$. Observe that the rule (iv) includes the additional binary predicate $E$. This predicate has a technical purpose and is used to mark which binary predicates $R$ are introduced by rules of the form (iv) (as opposed to rules from $\axioms$) during the chase. The usefulness of the predicate $E$ will become apparent in the subsequent section. Last, rule (v) expresses that the satisfaction of contradictory formulae implies a contradiction, which we denote by the special nullary predicate $\predcontr$. We interpret bodies and heads of rules from $\rsp \cup \axioms$ accordingly:

\begin{definition} Let $\pdmodel = (W,R,V)$ be a $\pdset$-model with $\assign : \ourterms \to W$ a partial function called an \emph{assignment}, and $X$, $Y$ formulae over the signature $\sig$. 
Then, 
\begin{itemize}

\item $\pdmodel,\assign\not\mermodels \predcontr$; 

\item $\pdmodel,\assign\mermodels \ppred_{\psi}(t)$ \iffi $\pdmodel, \assign(t) \mforcepd \psi$; 
    
\item $\pdmodel,\assign\mermodels \fpred(t,s)$ \iffi $(\assign(t),\assign(s)) \in R$ for $\fpred \in \set{\rpred, \epred}$;

\item $\pdmodel,\assign\mermodels X {\land} Y$ \iffi $\pdmodel,\assign \mermodels X$ and $\pdmodel,\assign \mermodels Y$; 

\item $\pdmodel,\assign \mermodels X {\lor} Y$ \iffi $\pdmodel,\assign \mermodels X$ or $\pdmodel,\assign \mermodels Y$.
    


\end{itemize}
If there exists an assignment $\assign$ such that $\pdmodel,\assign \mermodels X$, then we write $\pdmodel \mermodels X$.
\end{definition}

We now demonstrate that when the chase is run with $\rsp \cup \axioms$ over the database $\dbz = \set{\ppred_{\phi}(a)}$ with $a$ a fixed constant, the resulting set $\chases$ witnesses the unsatisfiability of $\phi$ if for each $\ch \in \chases$, $\predcontr \in \ch$, and witnesses the satisfiability of $\phi$ otherwise. This soundness (\cref{thm:sound}) and completeness (\cref{thm:complete}) result is proven below.

\begin{figure*}[t]
    \centering
\includegraphics[width=0.7\linewidth]{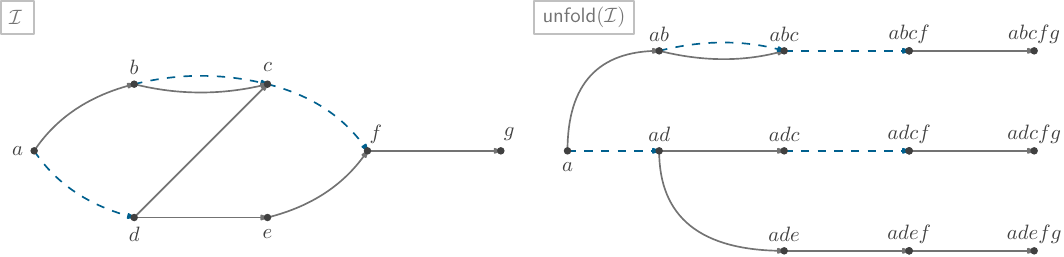}
    \caption{Example unfolding of a rooted DAG $\inst$. The dashed blue edges are $\rpred$, while solid gray edges are $\epred$. Note how terms in the unfolded instance are uniquely identified by paths from the root of $\inst$.}
    \label{fig:unfold}
\end{figure*}

\begin{lemma}\label{lem:sound-step} Let $\pdmodel = (\modw, \modr, \modv)$ be a $\pdset$-model and $\erule \in \rsa \cup \rsp$. If $\pdmodel \mermodels \body{\erule}$, then $\pdmodel \mermodels \head{\erule}$.
\end{lemma}

\begin{proof} We prove the result by a case distinction on $\erule \in \rsa \cup \rsp$. We consider the cases where $\erule$ is of form (iv) in \cref{def:rsp} and the case where $\erule \in \pdset$, and note that the remaining cases are similar.

First, let $\erule = \narule{\ppred_{\mldia\psi}(x)}{y}{\epred(x,y) \land \rpred(x, y) \land \ppred_{\psi}(y)}$ and suppose $\pdmodel, \assign \mermodels \ppred_{\mldia\psi}(x)$ with $\assign$ an assignment. Then, $\pdmodel, \assign(x) \mforcepd \mldia\psi$, implying the existence of a world $u \in W$ such that $(\assign(x),u) \in R$ and $\pdmodel, u \mforcepd \psi$. Let us define $\assign'$ such that $\assign'(t) = \assign(t)$ if $t \neq y$ and $\assign'(y) = u$. Then, it follows that $\pdmodel, \assign' \mermodels \epred(x,y) \land \rpred(x, y) \land \ppred_{\psi}(y)$.

For the second case, suppose $\erule = \forall x,y (\modr^{k}(x,y) \rightarrow \modr^{n}(x,y))$ with $\pdmodel,\assign \mermodels \modr^{k}(x,y)$ with $\assign$ an assignment. Since $\pdmodel$ is a $\pdset$-model, it immediately follows that there exist worlds $u_{1}, \ldots, u_{k-1} \in W$ such that $(\assign(x),u_{1}), \ldots, (u_{k-1},\assign(y)) \in R$. Let $z_{1}, \ldots, z_{n-1}$ be the existential variables in $\modr^{n}(x,y)$, and let us define $\assign'(t) = \assign(t)$ if $t \not\in \{u_{1}, \ldots, u_{n-1}\}$ with $\assign'(z_{i}) = u_{i}$ otherwise. Then, it is clear that $\pdmodel, \assign' \mermodels \modr^{n}(x,y)$.
\end{proof}

\begin{theorem}[Soundness]\label{thm:sound} If for each $\ch \in \chases$, $\predcontr \in \ch$, then $\mforcepd \negnnf{\phi}$.
\end{theorem}

\begin{proof} Suppose for a contradiction that $\predcontr \in \ch$ for each $\ch \in \chases$, but $\not\mforcepd \negnnf{\phi}$. Then, there exists a $\pdset$-model $\pdmodel$ such that $\pdmodel \mermodels \ppred_{\phi}(a)$. By \cref{lem:sound-step} and the fact that each step of $\chases$ only applies rules from $\rsa \cup \rsp$, we know that $\pdmodel \mermodels \predcontr$, giving a contradiction, meaning $\mforcepd \negnnf{\phi}$.
\end{proof}

\begin{theorem}[Completeness]\label{thm:complete} If $\mforcepd \negnnf{\phi}$, then for each $\ch \in \chases$, $\predcontr \in \ch$.
\end{theorem}

\begin{proof} We prove the result by contraposition and assume that there exists a $\ch \in \chases$ such that $\predcontr \not\in \ch$. We now define a $\pdset$-model $\pdmodel = (W,R,V)$ from $\ch$ such that $W$ contains all terms in $\ch$, $(t,s) \in R$ \iffi $R(t,s) \in \ch$, and $t \in V(p)$ \iffi $\ppred_{p}(t) \in \ch$. Let us now verify that $\pdmodel$ is indeed a $\pdset$-model. First, observe that $W \neq \emptyset$ as it contains the constant $a$. Second, observe that by the definition of $\chases$, we have that $R$ satisfies each QDP in $\pdset$. Last, we note that $V$ is well-defined, i.e. it cannot be the case that $t \in V(p)$ and $t \not\in V(p)$ since then both $\ppred_{p}(t), \ppred_{\neg p}(t) \in \ch$, meaning rule $\nadrule{\ppred_{\psi}(x) \land \ppred_{\negnnf{\psi}}(x)}{\predcontr}$ would be applied at some step of the chase, ensuring that $\predcontr \in \ch$. 

We now show by induction on the complexity of $\psi$ that if $\ppred_{\psi}(t) \in \ch$, then $\pdmodel, t \mermodels \psi$. We argue the case where $\psi = \dia \chi$ as the remaining cases are similar. 
 If $\ppred_{\dia \chi}(t) \in \ch$, then at some step of the chase, $\narule{\ppred_{\mldia\chi}(x)}{y}{\epred(x,y) \land \rpred(x, y) \land \ppred_{\chi}(y)}$ will be applied, ensuring that $\epred(t, y), \ppred_{\chi}(y) \in \ch$ with $y$ fresh. By definition, $(t,y) \in R$, and by IH, $\pdmodel, y \mforcepd \chi$, implying $\pdmodel, t \mermodels \dia \chi$.

Last, as $\ppred_{\phi}(a) \in \ch$ by definition, it follows that $\pdmodel, a \mforcepd \phi$, showing that $\not\mforcepd \negnnf{\phi}$.
\end{proof}

Last, observe that every step of a fair derivation sequence of $\dbz$ and $\rsp \cup \rsa$ starts with $\dbz$ and successively `grows' a DAG rooted at $a$. As a consequence, the following observations hold, which is used in the next section.

\begin{observation}\label{cor:chase-is-dag}\label{obs:chase-has-minimal} Every element of $\chase{\dbz, \rsp \cup \rsa}$ is a DAG rooted at $a$.
\end{observation}

\begin{observation}\label{obs:r-paths} The maximal length of $\rpred$ paths in any element of $\chase{\dbz, \rsp \cup \rsa}$ is bound by $\md{\phi}$.
\end{observation}

\begin{proof} Note that creating $\rpred$ paths during the chase is possible only using the (\ref{rule:box}) rules and that their usage decreases the modal depth of the subscript modal formula by one. In addition, when a (\ref{rule:dia}) rule is applied the modal depth of the `propagated' formula $\psi$ is one less than $\dia \psi$. As the (\ref{rule:head-conjunction}) and (\ref{rule:head-disjunction}) rules only decompose conjunctions and disjunctions, the (\ref{rule:contr}) rules only introduce the nullary predicate $\predcontr$, and the $\rsa$ rules only introduce $R$ paths (leaving $E$ atoms unaffected), it follows that $\rpred$ paths must be bounded by $n = \md{\phi}$.
\end{proof}

%% file: 06b-ExistentialRules.tex
\section{Decidability via Templates}\label{sec:erules}

The goal of this section is to establish \cref{lem:existential-decidability-problem} (shown below), which, in conjunction with Theorems \ref{thm:sound} and \ref{thm:complete}, implies our main decidability result (Theorem~\ref{thm:main}).  

\begin{lemma}\label{lem:existential-decidability-problem} Given a set $\pdset$ of QDPs and a modal formula $\phi$, it is decidable to check if there exists a $\ch \in \chases$ such that $\predcontr \not\in \ch$.
\end{lemma}

To demonstrate the above lemma, we first show that every element of $\chases$ can be unfolded into a multi-tree (see \cref{def:unfold}), trimmed to a certain depth (see \cref{def:trim}), and simplified via the $\core$ operation (see \cref{sec:prelims}), yielding a finite instance called a template (\cref{def:map}, and \cref{obs:chase-to-map}). Second, we show that the existence of a template without $\predcontr$ is sufficient to conclude the existence of a $\ch \in \chases$ such that $\predcontr \not\in \ch$ (\cref{lem:map-to-chase}). In other words, we establish a correspondence between templates (of which there can be only finitely many) and members of $\chases$ (\cref{cor:map-equiv-chase}), culminating in a proof of the above lemma. As the size of templates depends on the modal depth of $\phi$ and elements of $\axioms$, we fix a few parameters that will be instrumental to conclude our proofs. In particular, we let $n = \md{\theform}$, and $N := 4n + K$ such that $K := \max\{k_{+} \ | \ k \rightarrow k_{+} \in \axioms\}$ for the remainder of the section.

\subsection{Proof of Lemma~\ref{lem:existential-decidability-problem}}

Let us first define three crucial operations -- $\unfold$, $\trim$, and $\unravel$ -- which will be used in our proofs. To increase comprehensibility, we provide examples of various operations we consider, which are included in~\cref{fig:unfold,fig:trim,fig:embedding-compliance}.

\begin{definition}[$\unfold$]\label{def:unfold}
Given an instance $\inst$, one of its terms $t$ and a term $s \in \ourterms$, we inductively define $\unfold(\inst, t, s)$ as:
\begin{equation*}
     \labels(t)[t \mapsto s \sbullet t] \;\cup \bigcup_{ \spred(t, r) \in \inst} \spred(s \sbullet t, s \sbullet t\sbullet r) \;\cup \!\!\!\!\!\!\!\!\! \bigcup_{r \text{ is a successor of } t} \!\!\!\!\!\!\!\!\! \unfold(\inst, r, s \sbullet t) 
\end{equation*}
for $\spred \in \set{\epred, \rpred}$ and where the operation $\jnst[t \mapsto s]$ replaces each occurrence of $t$ in an instance $\jnst$ with $s$. We will write $\unfold(\inst)$ instead of $\unfold(\inst, r, \varepsilon)$ if $\inst$ is a DAG rooted at $r$.
\end{definition}
\noindent
An example of the $\unfold$ operation can be found in \cref{fig:unfold}.

\begin{observation}\label{obs:unfold-preserves-hom}
    Given two rooted DAGs $\inst$ and $\jnst$, if there exists a homomorphism $h$ from $\inst$ to $\jnst$ then there exists a homomorphism $h'$ from $\unfold(\inst)$ to $\unfold(\jnst)$ 
\end{observation}

\begin{proof} Let $h'(w) := h(t_1) \sbullet h(t_2) \cdots h(t_n)$ for each term $w = t_1 \sbullet t_2 \cdots t_n$ in $\unfold(\inst)$. We know that $h'(w)$ exists since in $\inst$ there exists a path $P(t_{1},t_{n})$ such that $\word(P(t_{1},t_{n})) \sbullet t_n = w$ (recall that $\word$ omits the last element of its path).
 The path $P(t_{1},t_{n})$ has to be mapped to the path $h(P(t_{1},t_{n}))$ as $h$ is a homomorphism from $\inst$ to $\jnst$. Thus, $h(P(t_{1},t_{n}))$ witnesses the existence of the term $h'(w)$ in $\jnst$. Next, let us show that $h'$ is indeed a homomorphism. First, observe that if $\apred(w\sbullet t) \in \unfold(\inst)$, then $\apred(t) \in \inst$, meaning $\apred(h(t)) \in \jnst$, and so, $\apred(h'(w\sbullet t)) \in \unfold(\jnst)$. Second, let $\spred(w, w\sbullet t)$ be an atom of $\unfold(\inst)$. Then, $\spred(s, t)$ is an atom of $\inst$ where $w = w' \sbullet s$. Since $h(\spred(s, t))$ is an atom of $\jnst$, $\spred(h'(w), h'(w \sbullet t)) = \spred(h'(w') \sbullet h(s),\; h'(w') \sbullet h(s) \sbullet h(t))$ will be an atom of $\unfold(\jnst)$.
\end{proof}

\begin{definition}[$\trim$]\label{def:trim}
Given a multi-tree $\tree$ over $\set{\epred, \rpred}$ and a natural number $i$, we define $\trim_i(\tree)$ as the multi-tree $\tree$ with all the elements past depth $i$ removed, except for the elements reachable with $E$-paths from elements at depths no larger than $i$.
\end{definition}

An example of how $\trim$ works can be found in \cref{fig:trim}. The following observation is straightforward to prove using the above definition.

\begin{observation}\label{obs:trim-preserves-hom} Given two multi-trees $\tree$ and $\tree'$, if there exists a homomorphism $h$ from $\tree$ to $\tree'$ then there exists a homomorphism $h'$ from $\trim_i(\tree)$ to $\trim_j(\tree')$ for any $i,j \in \mathbb{N}$ such that $i \leq j$. 
\end{observation}

From \cref{cor:chase-is-dag} and \cref{obs:r-paths}, we get the following corollary:

\begin{corollary}\label{cor:trim-bound-r-paths}
The depth of $\trim_i \circ \unfold(\ch)$ for any element $\ch$ of $\chase{\dbz, \rsp \cup \rsa}$ is bounded by $i + n$. Thus, $\core \circ \trim_i \circ \unfold(\ch)$ is well-defined  for any $i \in \mathbb{N}$.
\end{corollary}

\begin{definition}[$\unravel$] We define $\unravel_i := \core \circ \trim_i \circ \unfold$. Note that $\unravel_i$ is idempotent for any $i \in \mathbb{N}$, that is, for any instance $\inst$, $(\unravel_i \circ \unravel_i)(\inst)$ is isomorphic to $\unravel_i(\inst)$.
\end{definition}

\begin{remark} We assume that the $\unfold$, $\trim$, and $\unravel$ operations preserve the nullary atom $\predcontr$ if it appears in the input instance.
\end{remark}

\begin{figure*}
    \centering
\includegraphics[width=0.7\linewidth]{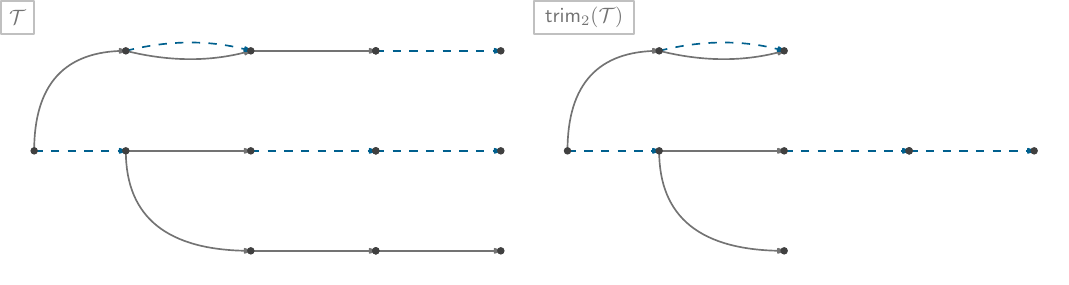}
    \caption{Example trimming of a multi-tree to the depth of two. The dashed blue edges are $\rpred$, while solid gray edges are $\epred$. Note how $\rpred$ paths are preserved if they start below depth of three.}
    \label{fig:trim}
\end{figure*}

\begin{observation}\label{obs:unravel-preserves-rsp}
For any instance $\inst$ and $i \in \mathbb{N}$, if $\inst \models \rsp$, then $\unravel_i(\inst) \models \rsp$.
\end{observation}

\begin{proof} The observation holds for rules of the form (\ref{rule:head-conjunction}), (\ref{rule:head-disjunction}), (\ref{rule:box}), and (\ref{rule:contr}) as every operation composed into $\unravel$ preserves the satisfaction of such rules. For a rule $\erule$ of the form (\ref{rule:dia}), $\unfold$ trivially preserves the satisfaction of $\erule$, $\trim$ never removes witnesses necessary for the satisfaction of the rule, and $\core$ also trivially preserves the satisfaction of the rule. Hence, such a rule will be satisfied as well.
\end{proof}

\begin{observation}\label{obs:unravel-preserves-hom} Given two rooted DAGs $\inst$ and $\jnst$, if there exists a homomorphism $h$ from $\inst$ to $\jnst$, then there exists a homomorphism $h'$ from $\unravel_i(\inst)$ to $\unravel_j(\jnst)$ for $i,j \in \mathbb{N}$ such that $i \leq j$.
\end{observation}

\begin{proof} From \cref{obs:unfold-preserves-hom}, and \cref{obs:trim-preserves-hom} we have that there exists homomorphism $h''$ from $\trim_i \circ \unfold(\inst)$ to $\trim_j \circ \unfold(\jnst)$. Note that $\core \circ \trim_i \circ \unfold(\inst)$ is a sub-instance of $\trim_i \circ \unfold(\inst)$. Thus, there exists a homomorphism $f$ from the former to $\trim_j \circ \unfold(\jnst)$ that is a restriction of $h''$. Moreover, by definition, there exists a homomorphism $g$ from $\trim_j \circ \unfold(\jnst)$ to $\core \circ \trim_j \circ \unfold(\jnst)$. Therefore, $g \circ f$ is the required homomorphism.
\end{proof}


\begin{definition}[$\embed$]\label{def:embed} Given a multi-tree $\tree$ and two of its terms $t'$ and $t_+$ such that there is a partial isomorphism $f$ from $\tree_{t_+}$ to $\tree_{t'}$ we write $\embed(\tree, t', t_+, f)$ to indicate an instance $\bar{f}(\tree)$ where $\bar{f}$ is an extension of $f$ over the set of all terms that is the identity function outside the domain of $f$. If such $f$ exists we say that $t_+$ is \emph{embeddable} in $t'$. We say that embedding is \emph{proper} \iffi $\unravel_{\ell} \circ \embed(\tree, t', t_+, f)$ is isomorphic to $\tree$, where $\ell$ is the maximal length of a path in $\tree$.
\end{definition}

An example of how $\embed$ works can be found in \cref{fig:embedding-compliance}.

\begin{definition}[compliance]\label{def:compliance} Given a QDP $\kkp$, a multi-tree $\tree$ and one of its terms $t$, we say that $t$ is $(\kkp)$-\emph{compliant} \iffi for every descendant term $t'$ of $t$ at a distance $k$ there exists a descendant term $t_+$ of $t$ at a distance $\kp$ such that $t_+$ is properly embeddable in $t'$ through the partial isomorphism $f$. We call every such $\pair{t', t_+, f}$ a \emph{compliance witness} for $t$. We extend this notion to sets of QDPs in the natural way.
\end{definition}

Observe that the right side of \cref{fig:embedding-compliance} shows that $s'$ and $s_+$, along with the partial isomorphism depicted as the purple arrows, serves as a $(2 \rightarrow 3)$-compliance witness for $s$ and $s'$.

\begin{observation}\label{obs:unravel-is-compliant}
Given a rooted DAG $\inst$ that is a model of $\rsa$ we have that every term $t$ of $\unravel_N(\inst)$ at a depth no greater than $2n$ is $\rsa$-compliant.
\end{observation}

\begin{proof} We prove that $\unravel_N(\inst) = \core \circ \trim_N \circ \unfold$ preserves the compliance requirement by arguing that each of the composed operations preserves it. First, $\unfold(\inst)$ trivially satisfies the compliance requirement. Second, assume toward a contradiction that the result $\jnst = (\trim_N \circ \unfold)(\inst)$ does not satisfy that requirement. Let $t$ and $t'$ be two terms of $\jnst$ that are missing a compliance witness for $\kkp \in \rsa$. Let $\pair{t', t_+, f}$ be the witness for $t$ and $t'$ from $\unfold(\inst)$. Observe that $\pair{t', t_+, f}$ cannot be used as a compliance witness for $t$ in $\jnst$ since there exists a term $s_+$ that is descendant $t_+$ in $\jnst$ such that $f(s_+)$ is not a term of $\jnst$. This, however, is a contradiction for the following three reasons: (1) Both $t'$ and $t_+$ are terms of $\jnst$ as the depth of $t$ is $\leq 2n$, and thus, the depths of $t'$ and $t_+$ do not exceed $2n + K \leq N$. (2) The paths in $\unfold(\inst)$ from $t_+$ to $s_+$ and from $t'$ to $f(s_+)$ are isomorphic. (3) The depth of $s_+$ in $\unfold(\inst)$ is $\kp - k$ greater than that of $f(s_+)$.

Note that if $s_+$ is reachable by an $\rpred$ path from $t_+$, then the same can be said about the path from $t'$ to $f(s_+)$. Therefore, $\trim_N$ should preserve $f(s_+)$ as $t'$ is at depth smaller than $N$ in $\unfold(\inst)$. If $s_+$ is not reachable from $t_+$ by an $\rpred$-path, it means the same for the path from $t'$ to $f(s_+)$; however, as the depth of $s_+$ is greater than that of $f(s_+)$ in $\unfold(\inst)$ we we know that $s_+$ should be removed by $\trim_N$ as $f(s_+)$ is. Finally, observe that $\core(\jnst)$ trivially satisfies the compliance requirement.
\end{proof}

Enough groundwork has been laid to define the notion of a template. We encourage the reader to briefly skim through the rest of this section to understand how templates will be used to prove Lemma~\ref{lem:existential-decidability-problem} as this justifies the imposition of the four properties mentioned in the definition below.


\begin{definition}\label{def:map}
    A \emph{template} $\map$ is a multi-tree that satisfies the following properties:
    \begin{enumerate}
        \item $\unravel_N(\map)$ is isomorphic to $\map$. \label{map:trim-id-on-map}\label{map:cored}
        \item Every term of $\map$ at a depth no greater than $2n$ is $\axioms$-compliant.\label{map:shallow-compliance}
        \item $\map$ is a model of $\rsp$. \label{map:modelhood}
        \item The root of $\map$ is labeled with $\ppred_{\theform}$. \label{map:phi-at-root}
    \end{enumerate}
\end{definition}

The core idea behind a template is that it serves as both a relaxation of the chase (as witnessed by \cref{obs:chase-to-map} below), yet is sufficiently strong enough to 
encode at least one element of the chase (see \cref{lem:map-to-chase} below).


\begin{observation}\label{obs:chase-to-map}
For any $\ch \in \chases$, $\unravel_N(\ch)$ is a template. Moreover, if $\predcontr\not\in \ch$, then $\predcontr\not\in \unravel_N(\ch)$.
\end{observation}
\begin{proof} We argue that $\unravel_N(\ch)$ satisfies properties (\ref{map:trim-id-on-map}) -- (\ref{map:phi-at-root}) of \cref{def:map}. Observe that (\ref{map:trim-id-on-map}) follows from the idempotence of $\unravel_N$, (\ref{map:shallow-compliance}) follows from~\cref{obs:unravel-is-compliant}, (\ref{map:modelhood}) follows from~\cref{obs:unravel-preserves-rsp}, and (\ref{map:phi-at-root}) is trivial. As $\unravel$ does not introduce new nullary atoms, the last assertion is satisfied as well.
\end{proof}

In essence, the above observation tells us that if there exists a $\ch \in \chases$ such that $\predcontr \not\in \ch$, then there exists a template witnessing this fact. To complete our decidability argument however, we will also need to establish the converse of this, i.e. the lemma below. Importantly, the definition of a template has been tailored so that the proof of this more difficult direction goes through.


\begin{lemma}\label{lem:map-to-chase} If $\map$ is a template such that $\predcontr\not\in \map$, then there exists a $\ch \in \chases$ such that $\predcontr\not\in \ch$.
\end{lemma}

\begin{figure*}
    \centering
    \includegraphics[width=0.7\linewidth]{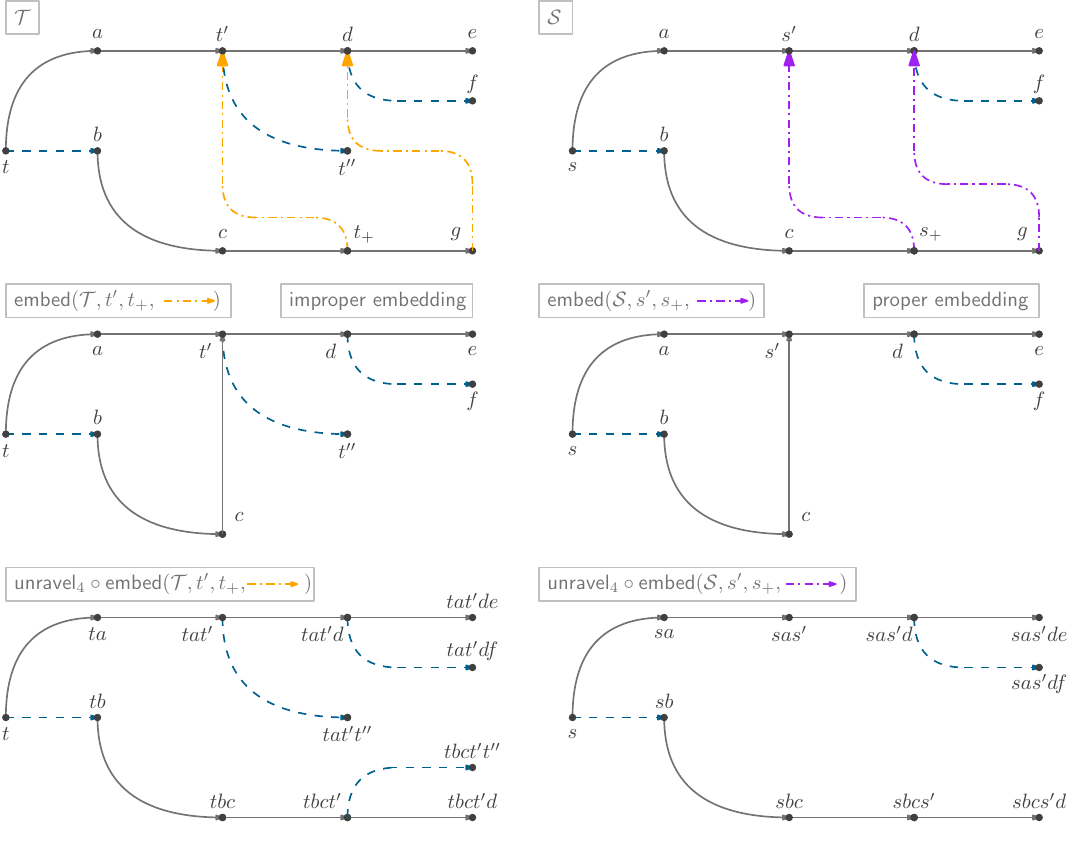}
    \caption{Two embeddings along with their unravelings where dash-dotted lines represent partial isomorphisms. The left embedding is improper as witnessed by an additional outgoing $\rpred$-atom of $t_+$. The right embedding is proper, as witnessed by its unraveling being isomorphic to $\mathcal{S}$. Note that in both cases $t_+$ and $s_+$ are deeper then the counterparts $t'$ and $s'$, and thus the sub-trees of $tbct'$ and $sbcs'$ are only partially isomorphic to the respective sub-trees $tat'$ and $sas'$.}
    \label{fig:embedding-compliance}
\end{figure*}

We shall prove the above in the next section. Before that, we show how the lemma fits into the scheme of our decidability proof. From Observation~\ref{obs:chase-to-map} and Lemma~\ref{lem:map-to-chase} above we have:

\begin{corollary}\label{cor:map-equiv-chase} There exists a template $\map$ such that $\predcontr\not\in \map$ \iffi there exists a $\ch \in \chases$ such that $\predcontr\not\in \ch$.
\end{corollary}

One can leverage the above to decide whether $\chases \models \predcontr$ by enumerating all the templates (of which there are only a finite  number to consider) and checking whether there exists one that does not contain $\predcontr$. From this we get Lemma~\ref{lem:existential-decidability-problem}, and thus, the decidability result (Theorem~\ref{thm:main}).
The complexity of this decision procedure is straightforward to determine:
\begin{enumerate}
    \item Checking whether a multi-tree is a template is trivially in \textsc{PSpace} with respect to the size of the input multi-tree.
    \item The size of a template is at most exponential in size of the modal formula $\theform$ and $\rsa$ as template branching is bounded by $2^n$ and its depth is at most $N + n$.
\end{enumerate}
Therefore, one can iterate in \textsc{\algcomplexity} through all multi-trees of exponential size and check whether they are templates or not.

\subsection{Proof of Lemma~\ref{lem:map-to-chase}}

\begin{definition}\label{def:following-map}
    Given a DAG $\inst$ rooted at $a$ and a template $\map$ we say that $\inst$ \emph{follows} $\map$ through homomorphism $h$ \iffi the following hold:
    \begin{enumerate}
        \item $h$ is a homomorphism from $\unravel_n(\inst)$ to $\map$.
        \item For every term $t$ of $\inst$ and all words $w, w'$, if $wt$ and $w't$ are terms of $\unravel_n(\inst)$, then $\labels(h(wt)) = \labels(h(w't))$.
    \end{enumerate}
\end{definition}

\noindent
Note that in the above definition we use $\unravel_n$ as opposed to $\unravel_N$. The following is the main result of this section:
\begin{lemma}\label{lem:main-technical}
  For every instance $\inst$, $(\rsp \cup \rsa)$-trigger $\pi$ in $\inst$, and template $\map$, if $\inst$ follows $\map$, then there exists an $\inst' \in \apply{\inst, \pi}$ such that it follows $\map$.
\end{lemma}
\begin{proof} Let $\pi = \pair{\erule, f}$, and $h$ be the homomorphism from $\unravel_n(\inst)$ to $\map$. We make a case distinction, assuming $\erule\in \rsp$ first, and assuming $\erule\in \rsa$ second.

Let $\erule\in \rsp$ and suppose $\erule$ is non-disjunctive. As the template $\map$ is a model of $\rsp$ there are no active triggers in it for any rule of $\rsp$. Thus, $\chase{\map, \rsp} = \set{\map}$. The rest follows from Observation~\ref{obs:chase-preserves-homs}, as the instance in the singleton $\apply{\inst, \pi}$ can follow the template $\map$ simply through $h$. If $\erule\in \rsp$ is disjunctive, then $\erule$ is of the form (\ref{rule:head-disjunction}) in \cref{def:rsp}. Let $t = f(x)$ for $x$ as in (\ref{rule:head-disjunction}). From the fact that $\inst$ follows the template $\map$ we know that all the copies of $t$ that were not removed by $\trim$ or $\core$ are mapped to the set $U$ of terms of $\map$ sharing the same set of labels (by Definition~\ref{def:following-map}). As $\map$ is a model of $\rsp$ we can conclude that there exists a symbol $\spred$ from $\sig$ such that $\spred(u) \in \map$ for each element $u$ of $U$ and that $\spred(x)$ is a disjunct in $\head{\erule}$. From this we can set $\inst'$ to be $\inst \cup \set{\spred(t)}$. We then have that $\inst'$ follows $\map$ through $h$.

Let us now assume that $\erule = k \rightarrow \kp \in \rsa$. Let $\inst'$ be the result of applying $\pi$ to $\inst$, $P(t,t')$ be the path from some $t$ to $t'$ witnessing that $\pi$ was not satisfied in $\inst$ and $P_+(t,t')$ be the freshly created path in $\inst'$. Let $s$ denote term $wt$ of $\unravel_n(\inst)$ for some word $w$. Let $s' = w\sbullet\word(P)\sbullet t'$ and $s_+ = w\sbullet\word(P_+(t,t'))\sbullet t'$. We consider two cases depending on whether $s' \in \unravel_n(\inst')$.
    
Assume $s' \not\in \unravel_n(\inst')$. This case is trivial. Note that elements created by $\unfold$ that appear on the copy of $P_+$ 
    can be homomorphically mapped to respective elements of the copy of $P$. Note that the remaining part of the copy of $P_+$ is no longer than that of $P$. Therefore, those elements can be ``removed'' by the $\core$ operation and will not have to be homomorphically mapped to $\map$. 
    
Assume $s' \in \unravel_n(\inst')$. Note that from \cref{cor:trim-bound-r-paths} we have $|w| + |P| < 2n$. Thus we can use the Property~\ref{map:shallow-compliance} of template $\map$ for its nodes $u = h(s)$ and $u' = h(s')$. Let $\pair{u', u_+, g}$ be the $\kkp$ compliance witness for $u$. We can simply extend $h$ to map the sub-tree $\tree_{s_+}$ of $s_+$ to the sub-tree $\tree_{u_+}$ of $u_+$ and the path from $s$ to $s_+$ to map to the path from $u$ to $u_+$. Note that the path from $s$ to $s_+$ contains only terms that have no labels (disregarding $s$ and $s_+$). Observe that the sub-tree $\tree_{s_+}$ is by definition equal to $\core \circ \trim_{n - (\kp - k)}(\tree_{s'})$ where $\tree_{s'}$ is the sub-tree rooted at $s'$, and that $\tree_{u_+}$ is simply equal to $\core \circ \trim_{N - (\kp - k)}(\tree_{u'})$ where $\tree_{u'}$ is the sub-tree rooted at $u'$. This follows from the fact that the embedding from $u_+$ to $u'$ is proper. As there exists a homomorphism from $\tree_{s'}$ to $\tree_{u'}$ by \cref{obs:unravel-preserves-hom} we know there exists a homomorphism from $\tree_{s_+}$ to $\tree_{u_+}$ since $N > n$.

All that remains is to show that $\inst'$ and such extended $h$ adheres to the second point of Definition~\ref{def:following-map}. Let $S$ be the set of all copies of $t$ appearing in $\unravel_n(\inst')$ that satisfy $s\sbullet\word(P) \in \unravel_n(\inst')$ for each $s \in S$, let $S'$ be the set $\set{s\sbullet\word(P) \mid s \in S}$, and $S_+$ be the set $\set{s\sbullet\word(P_+) \mid s \in S}$. By assumption, elements of $h(S')$ share the same labels, and thus, $h(S_+)$ share same set of labels as $h(S')$ due to how $h$ was extended using Property~\ref{map:shallow-compliance} of the template. This reasoning trivially extends to sub-trees of elements in $S'$ and $S_+$ 
\end{proof}

By inductively applying \cref{obs:chase-has-minimal} and \cref{lem:main-technical} we obtain:

\begin{lemma} Given a template $\map$ there exists an instance $\ch \in \chase{\dbz, \rsp \cup \rsa}$ such that $\unravel_n(\ch)$ homomorphically maps to $\map$.
\end{lemma}

Finally, as $\unravel$ preserves the nullary atom $\predcontr$, we obtain Lemma~\ref{lem:map-to-chase} as a corollary of the above.

%% file: 07-Conclusions.tex
\section{Concluding Remarks}\label{sec:conclusions}

In this paper, we provided a novel algorithm that uniformly decides quasi-dense logics -- a substantial subclass of `modal reduction logics' (i.e. logics axiomatized by extending $\logick$ with modal reduction principles). Our method relied on an established correspondence between members of the disjunctive chase and finite instances, called templates, which encode Kripke models witnessing the satisfiability of a modal formula. In the process, we provided a toolkit of model-theoretic operations which we aim to extend and adapt as required to approach the general decidability problem of modal reduction logics. We remark that our result is orthogonal to the other, notable class of modal reduction logics shown decidable by Zakharyaschev~\cite{Zak92}, and thus we have taken a meaningful step forward and improved our understanding of this longstanding open problem.

